\documentclass[conference]{IEEEtran}
\IEEEoverridecommandlockouts
\usepackage{cite}
\usepackage{dashbox}
\usepackage[
    n,
    operators,
    advantage,
    sets,
    adversary,
    landau,
    probability,
    notions,    
    logic,
    ff,
    mm,
    primitives,
    events,
    complexity,
    asymptotics,
    keys]{cryptocode}

\usepackage{url}
\usepackage{amsmath,amssymb,amsfonts}
\usepackage{glossaries-extra}
\usepackage{graphicx}
\usepackage{float}
\usepackage{subfig}
\usepackage{framed}
\usepackage{textcomp}
\usepackage{xcolor}
\def\BibTeX{{\rm B\kern-.05em{\sc i\kern-.025em b}\kern-.08em
    T\kern-.1667em\lower.7ex\hbox{E}\kern-.125emX}}

\usepackage{bm}
\usepackage{multirow}

\usepackage{amsthm}
\theoremstyle{definition}
\newtheorem{definition}{Definition}

\newtheorem{prop}{Proposition}
\newtheorem{lemma}{Lemma}

\usepackage{moresize}
\usepackage{algorithm}
\usepackage{algpseudocode}
\usepackage{threeparttable}
    
\begin{document}
\bstctlcite{IEEEexample:BSTcontrol}

\title{A Critical Look into Threshold Homomorphic Encryption for Private Average Aggregation}

\author{\IEEEauthorblockN{Miguel Morona-Mínguez, Alberto Pedrouzo-Ulloa and Fernando Pérez-González}
\IEEEauthorblockA{atlanTTic, Universidade de Vigo \\
\{mmorona, apedrouzo, fperez\}@gts.uvigo.es}
 \thanks{This is the author-submitted version (preprint) of a paper published in the Proceedings of the 2nd IEEE International Conference on Federated Learning Technologies and Applications (FLTA 2024). The final version is available in IEEE Xplore: \protect\url{https://doi.org/10.1109/FLTA63145.2024.10840167}.}
 }

\maketitle

\begin{abstract}
Threshold Homomorphic Encryption (Threshold HE) is a good fit for implementing private federated average aggregation, a key operation in Federated Learning (FL). Despite its potential, recent studies have shown that threshold schemes available in mainstream HE libraries can introduce unexpected security vulnerabilities if an adversary has access to a restricted decryption oracle. This oracle reflects the FL clients' capacity to collaboratively decrypt the aggregated result without knowing the secret key. This work surveys the use of threshold RLWE-based HE for federated average aggregation and examines the performance impact of using smudging noise with a large variance as a countermeasure. We provide a detailed comparison of threshold variants of $\mathsf{BFV}$ and $\mathsf{CKKS}$, finding that $\mathsf{CKKS}$-based aggregations perform comparably to $\mathsf{BFV}$-based solutions.
\end{abstract}

\begin{IEEEkeywords}
Machine Learning, Federated Learning, Private Aggregation, Homomorphic Encryption, Multiple Keys.
\end{IEEEkeywords}

\section{Introduction}

\label{sec:intro}
Federated Learning (FL) was introduced in~\cite{MMRHA17,KMABBBBCC21} as a Machine Learning (ML) setting where, under the coordination of a central server, multiple clients collaborate with the aim of solving a learning task. One distinguishing feature of FL relies on the fact that the data used for training is distributed among the different clients, in such a way that datasets remain separate and stored locally in each of the clients' premises.

The learning process (see Fig.~\ref{fig:flprotocol}) is as follows: Starting from an initial common ML model $\vec{w}^*$, an FL protocol typically consists in a series of rounds in which, (1) the clients locally update the common model by partially training with only their own individual datasets, obtaining a local update $\vec{w}^i$, and (2) the central server or aggregator, combines all received local updates from the clients to obtain a new common ML model, which is done by computing an aggregation function $\vec{w}^* = f_{\Sigma}(\vec{w}^1,\ldots,\vec{w}^L)$. This process can be repeated several times till a preestablished convergence criteria, which measures the performance of the trained ML model, is fulfilled.

\begin{figure}[h!]
    \centering
    \includegraphics[width = \linewidth]{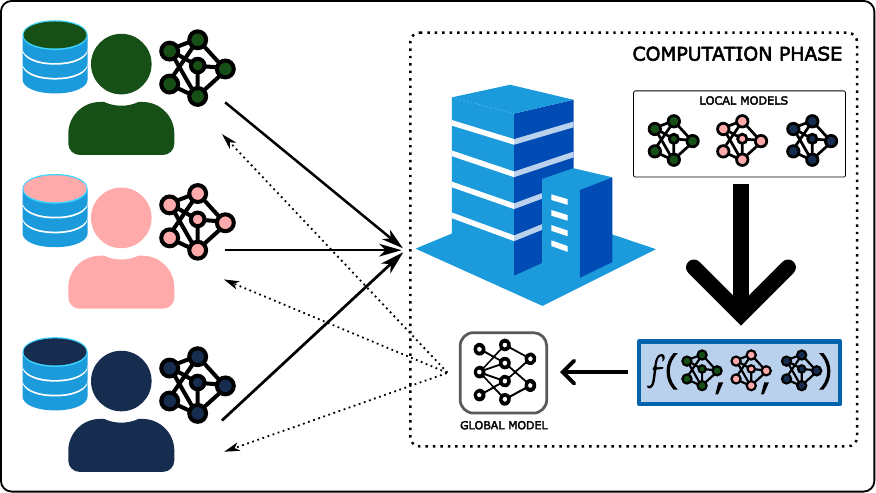}
    \caption{High-level description of the FL protocol.}
    \label{fig:flprotocol}
    \vspace{-0.4cm}
\end{figure}

\textbf{Privacy issues:} 
FL successfully avoids having all the training data gathered by the same centralized entity, subsequently reinforcing data control and privacy protection by isolating the local datasets in each client premises. Unfortunately, even though local data is not openly shared, if no further protection measures are applied, FL protocols still present significant privacy and security issues~\cite{MMRHA17}. Several works have shown how the exchanged clients' local ML updates can sometimes remember information of the data used during training, thus making them susceptible to various types of inference attacks, especially when adversaries have white-box access to both the exchanged local and aggregated updates~\cite{NSH19,LF20}.

By taking a closer look at the FL protocol, we observe that, even under a common security model as honest-but-curious with collusions and in contrast to the clients who only receive the aggregated ML update, the aggregator has direct access to all the individual updates sent by the clients. This additional information accessible to the aggregator aggravates the potential privacy risks posed by inference attacks, especially if the aggregator is compromised by an adversary. Consequently, this situation underlines the importance of incorporating adequate private aggregation solutions into the FL pipeline. As surveyed in~\cite{MOJC23}, a diverse number of technologies can be used for this purpose, although there is no definitive perfect candidate. 

One prominent example of those technologies corresponds to Homomorphic Encryption (HE) which, when extended to operate with multiple keys~\cite{MTBH21,AHSL22}, seems to precisely cover many of the requirements of private aggregation protocols.

\textbf{Scope of our work:} Recent research in the cryptographic community~\cite{CSBBC24,CCPSS24} has shown that mainstream HE schemes paired with multiple keys can introduce some unexpected security vulnerabilities
. In view of this privacy weakness, we focus here on how to correctly use modern HE schemes ($\mathsf{BFV}$~\cite{Brakerski12,FV12} for exact arithmetic and $\mathsf{CKKS}$~\cite{CKKS17} for approximate arithmetic) with multiple keys for the execution of private average aggregation. Our main interest lies on cross-silo FL settings, in which the ML model is built from stateful clients, who are always available and have enough computational power. Even so, we note that HE could also be adapted to work reasonably well in the cross-device setting if the clients, who are intermittently available, have still enough computational and memory resources for encryption/decryption.

\subsection{Main contributions}
We propose private aggregation methods based on the use of HE under multiple keys. Our solutions follow the blueprint of~\cite{MTBH21} making use of additive secret sharing to deal with secret keys. In particular, our main contribution is threefold:
\begin{itemize}
\item 
We survey the state of the art in threshold HE and its application for implementing efficient private aggregation protocols. In particular, we explore in detail the limitations of the commonly used $\mathsf{IND}$-$\mathsf{CPA}$ security model in characterizing the privacy vulnerabilities present in threshold HE when applied to FL protocols. We also discuss the main countermeasures to securely instantiate these schemes under the $\mathsf{IND}$-$\mathsf{CPA}^{\mathsf{D}}$ security model~\cite{LM21}.
\item Contrarily to what was previously established in~\cite{CSBBC24}, we show how, for the federated average rule, the $\mathsf{CKKS}$ scheme can still outperform $\mathsf{BFV}$ in terms of cipher expansion when a high smudging noise is used to guarantee privacy. We study and compare the range of parameters for which one scheme behaves better than the other.
\item We briefly discuss how, for the FL setting, the approximate nature of $\mathsf{CKKS}$ not only does not introduce additional cryptographic weaknesses as stated in~\cite{CSBBC24}, but also can bring some extra privacy protection benefits by naturally adding noise to the aggregated result. 
\end{itemize}

\subsection{Notation and structure}
Polynomials are denoted with regular lowercase letters, ignoring the polynomial variable (e.g., $a$ instead of $a(x)$) whenever there is no ambiguity. We also represent polynomials as column vectors of their coefficients $\vec{a}$. 
Let $R_q = \mathbb{Z}[x]/(1 + x^n)$, with $q \in \mathbb{N}$ and $n$ a power-of-two, be the polynomial ring in the variable $x$ reduced modulo $1 + x^n$ with coefficients belonging to $\mathbb{Z}_q$, but using the set of representatives $(-q/2, q/2]$. 
Let $\chi$ be the error distribution over $R_q$, whose coefficients are independently sampled from a discrete Gaussian with standard deviation $\sigma$ and truncated support over $[-B, B]$. 
Then, for a finite set $\mathcal{S}$, $x\leftarrow$ $\mathcal{S}$ denotes sampling $x$ uniformly at random from $\mathcal{S}$, while $x\leftarrow$ $\chi$ denotes sampling from the distribution $\chi$. Finally, both 
$||\vec{a}||$ and $||a||$ refer to the infinity norm of $||\vec{a}||$.

The rest of the document is organized as follows: Section~\ref{sec:prelim} briefly reviews the state of the art of HE and its application for private aggregation. We pay special attention to the privacy vulnerabilities present in threshold HE when the adversary has access to a restricted decryption oracle. 
Section~\ref{sec:proposedpriagg} details the baseline $\mathsf{BFV}$ and $\mathsf{CKKS}$ schemes. It also introduces threshold HE as a natural extension 
dealing with multiple keys inside the FL setting. We focus on describing the use of the smudging noise as the main countermeasure against the attacks discussed in~\cite{CSBBC24,CCPSS24}. We study its effects and we provide analytical comparisons among both threshold $\mathsf{BFV}$ and threshold $\mathsf{CKKS}$. 
We also discuss possible advantages that the approximate nature of $\mathsf{CKKS}$ can bring into the FL setting. 
Finally, Section~\ref{sec:conc} concludes with 
some future work lines.


\section{Preliminaries: State of the art for HE}
\label{sec:prelim}

Even though the application of HE for private aggregation has already been demonstrated in several works~\cite{SCRCS11,BJL16,ZLXWYL20,MSMSGS21}, the use of more modern lattice-based HE schemes, such as $\mathsf{BFV}$, $\mathsf{BGV}$, and $\mathsf{CKKS}$~\cite{ACCDGGHH19}, is more recent. Specifically, these schemes are based on the hardness of the Ring Learning with Errors (RLWE) problem and offer significant efficiency advantages by naturally supporting SIMD (Single-Instruction Multiple-Data) operations, which allow packing several values together under the same ciphertext~\cite{SV14}. These features have motivated research into the combination of RLWE-based HE with other complementary privacy-related tools. For instance, some works explore their integration with verification techniques to address malicious servers~\cite{FGP14,ACGS23}, and with Differential Privacy (DP) to better manage a larger number of colluding dishonest clients in FL~\cite{SCSSG23}, among others.

\subsection{Limitations for FL deployment with HE} An advantage of HE-based aggregations is that they naturally align with the communication flow of the described FL protocol between clients and the aggregator. First, the clients encrypt their updates, then the aggregator homomorphically computes the aggregation function $f_{\Sigma}(\vec{w}^1,\ldots,\vec{w}^L)$, and finally, the clients can decrypt the aggregated update. However, these works assume a very simple deployment scenario in which all clients share the same secret and public keys, implicitly imposing a strong non-colluding assumption among the aggregator and those parties with access to the secret key.

\subsubsection*{HE and multiple keys}
To overcome this limitation, we can upgrade HE schemes to handle ciphertexts under multiple keys~\cite{AHSL22}. Several works propose methods to homomorphically aggregate ciphertexts encrypted under different keys~\cite{MNSL22,PBCSSZ23}, but the most versatile option corresponds to threshold variants of single-key HE, such as the MHE (Multiparty HE) scheme proposed in~\cite{MTBH21}. In this scheme, the secret key associated with the public key of a single-key HE scheme is divided into several additive shares, which are distributed among the clients. The secret shares are generated during a setup phase that is run offline before starting the FL protocol. Note that MHE enables encryption of the local updates under a common secret key that is unknown to any single party, requiring decryption to be done collaboratively by all clients.

In return for the need for a setup phase, threshold HE schemes scale very well with the number of clients and allow for easy reuse of results from single-key HE. For example, in~\cite{MNSL22,PBCSSZ23}, efficient homomorphic aggregations are limited to average aggregation. In contrast, with threshold schemes, we could utilize the results of~\cite{CGPSSZ23}, where the authors studied the homomorphic evaluation of several Byzantine-robust aggregation rules (e.g., the median). Unfortunately, in~\cite{CSBBC24,CCPSS24}, the authors demonstrate how threshold HE schemes are particularly vulnerable when there is a non-negligible probability of incorrect decryption. Both works exemplify how this vulnerability can be exploited to implement effective key-recovery attacks against several mainstream HE libraries.

\subsection{Security models for HE: From $\mathsf{IND}$-$\mathsf{CPA}$ to $\mathsf{IND}$-$\mathsf{CPA}^\mathsf{D}$}
Previous works~\cite{MTBH21} using threshold HE discussed how these schemes could be proven to be secure under the $\mathsf{IND}$-$\mathsf{CPA}$ (Indistinguishability under Chosen-Plaintext Attack) security model by relying on the difficulty of breaking the RLWE problem. Under this passive security model, the adversary only has access to encryptions by means of an encryption oracle. Note that this security model is sufficient from the aggregator's perspective, who only sees RLWE ciphertexts and the public key. However, the absence of a decryption oracle does not fully capture the peculiarities of the aggregation functionality, especially from the clients' perspective, who can collaboratively decrypt ciphertexts without knowing the secret key. The consequences of having a decryption oracle for HE schemes is precisely an issue which is being explored since the introduction of the $\mathsf{IND}$-$\mathsf{CPA}^\mathsf{D}$ security model in~\cite{LM21,LMSS22}. 

\section{An analysis of Private Aggregation with HE}
\label{sec:proposedpriagg}

This section surveys the available algorithms to homomorphically evaluate $f_\Sigma(\cdot)$ in an FL protocol. We start by giving a brief overview of how HE fits inside the FL protocol. Next, we introduce additive HE and describe how a threshold variant can be built on top of the single-key counterparts. We then focus on the $\mathsf{IND}$-$\mathsf{CPA}^\mathsf{D}$ security model and the use of smudging noise as a countermeasure against the key-recovery attacks proposed in~\cite{CSBBC24,CCPSS24}. Finally, we analyze the performance degradation caused by the large smudging noise, and provide analytical comparisons among different solutions for private aggregation based on threshold $\mathsf{BFV}$ and threshold $\mathsf{CKKS}$. We indicate for which protocol parameters one scheme outperforms the other.

\textbf{General Overview:} The incorporation of HE to securely compute the aggregation function inside the FL protocol obeys the following two principles: (1) unprotected local data is isolated in different silos and remains in the clients' premises, and (2) all local updates leaving a silo are previously encrypted with HE. This work follows the blueprint of the threshold Multiparty Homomorphic Encryption (MHE) scheme proposed in~\cite{MTBH21} (see Section~\ref{sec:mhe} for further details). Figure~\ref{fig:priflprotocol} sketches the process followed to train a model. Note that steps $2$-$5$ are typically repeated several times till convergence is achieved:
\begin{enumerate} \vspace{-0.09cm}
    \item \textit{Setup step}: All involved clients generate individual secret and a collective public keys. An initial common ML model is agreed among the Clients and the Aggregator.
    \item \textit{Local training step}: Each Client trains locally with its data a local ML update.
    \item \textit{Input step}: Each Client encrypts its local ML update to be outsourced under a collective public key.
    \item \textit{Evaluation step}: The aggregator will be in charge of homomorphically computing the functionality $f_\Sigma$.
    \item \textit{Output step}: Finally, all involved parties collaboratively decrypt the aggregated ML update.
\end{enumerate}

\begin{figure}[h!]
    \centering
    \includegraphics[width = \linewidth]{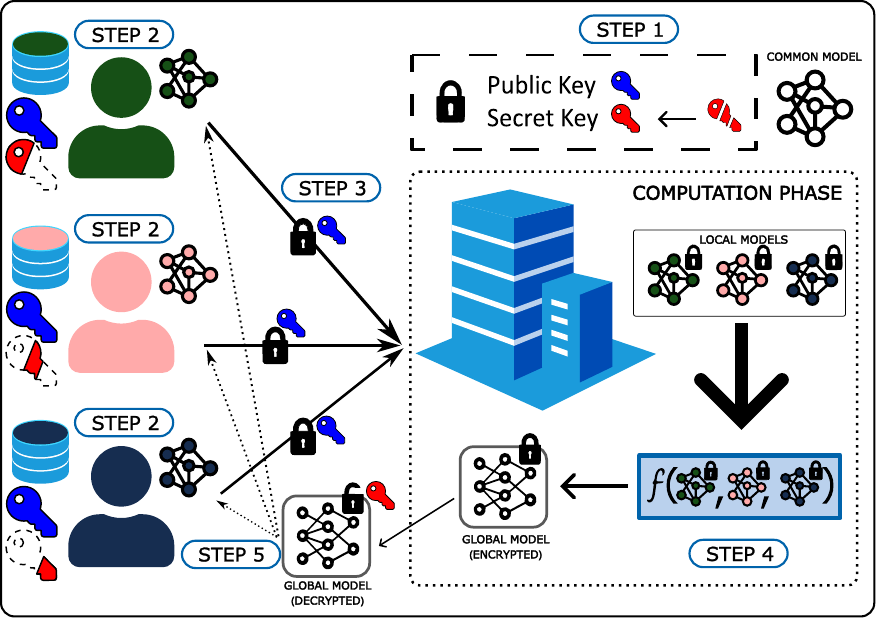}
    \caption{Protocol for private average aggregation.}
    \label{fig:priflprotocol}
    \vspace{-0.4cm}
\end{figure}

\subsection{Background: Single-Key HE}
\label{sec:background}
\textbf{}


As our interest in this work relies on the computation of average aggregations $\vec{w}^* = f_{\Sigma}(\vec{w}^1,\ldots,\vec{w}^L)$, we will restrict ourselves to additive HE schemes. For ease of exposition we will focus on the case of addition $f_{\Sigma}(\vec{w}^1,\ldots,\vec{w}^L) = \sum_i \vec{w}^i$.

\begin{definition}[Additive Homomorphic Encryption]\label{def:addhe}
    Let $\mathsf{E} = \{\mathsf{Setup}, \mathsf{SecKeyGen}, \mathsf{PubKeyGen}, \mathsf{Enc}, \mathsf{Dec}\}$ be an asymmetric encryption scheme, whose security is parameterized by $\lambda$. An additive homomorphic encryption scheme extends $\mathsf{E}$ with the $\mathsf{Add}$ procedure, having $\mathsf{E_{ahe}} = \{\mathsf{E}.\mathsf{Setup},$ $ \mathsf{E}.\mathsf{SecKeyGen}, \mathsf{E}.\mathsf{PubKeyGen}, \mathsf{E}.\mathsf{Enc}, \mathsf{E}.\mathsf{Dec}, \mathsf{Add}\}$, where:
    \begin{itemize}
        \item Let $c_{\vec{v}}$ and $c_{\vec{w}}$ be two ciphertexts encrypting, respectively, $\vec{v}$ and $\vec{w}$. The procedure $\mathsf{Add}$ satisfies:
        \begin{equation} \label{eq:correctness-add}
            \mathsf{E_{ahe}}.\mathsf{Dec}(\mathsf{E_{ahe}}.\mathsf{Add}(c_{\vec{v}}, c_{\vec{w}})) = \vec{v} + \vec{w}.
        \end{equation}
    \end{itemize}
\end{definition}



We have chosen two \textit{state-of-the-art} RLWE-based cryptosystems to instantiate the $\mathsf{E}_\mathsf{ahe}$ scheme which is informally introduced in Definition~\ref{def:addhe}. Specifically, they correspond to $\mathsf{BFV}$~\cite{Brakerski12,FV12} and $\mathsf{CKKS}$~\cite{CKKS17}, being each of them, respectively, a representative of an exact and an approximate HE scheme. We refer the reader to Table~\ref{tab:bfv-ckks} for more details on their primitives and parameters.

Both schemes allow us to perform SIMD-style additions by packing a total of $n$ values.\footnote{For the more general case of SIMD-style additions and multiplications, we can pack in each separate ciphertext a total of, respectively, $n$ and $n/2$ values for $\mathsf{BFV}$ and $\mathsf{CKKS}$.} However, they still present some significant differences regarding correctness~\cite{ABMP24}. While for $\mathsf{BFV}$ we can choose adequate parameters which satisfy Eq.~\eqref{eq:correctness-add} for all the homomorphic additions needed to compute $f_\Sigma(\cdot)$, capturing the behavior of approximate HE requires to relax the exact correctness of $\mathsf{E}$. This can be done by considering $\mathsf{E_{ahe}}.\mathsf{Dec}(\mathsf{E_{ahe}}.\mathsf{Add}(c_{\vec{v}}, c_{\vec{w}})) \approx \vec{v} + \vec{w}$ instead. Now, given a certain error margin $\epsilon > 0$, we say that the $\mathsf{CKKS}$ scheme is approximately correct, if the following holds:
\begin{equation*}
||f(\vec{w}^1,\ldots,\vec{w}^L) - \mathsf{E_{ahe}}.\mathsf{Dec}(f(c_{\vec{w}^1},\ldots,c_{\vec{w}^L})) || < \epsilon.
\end{equation*}



\textbf{Some details on the HE schemes:}
To ensure a fair comparison between $\mathsf{BFV}$ and $\mathsf{CKKS}$, we have made some minor modifications to the original $\mathsf{CKKS}$ scheme described in~\cite{CKKS17}. These changes primarily involve sampling the randomness in both $\mathsf{BFV}$ and $\mathsf{CKKS}$ from the same distributions. Additionally, $\mathsf{CKKS}$ and $\mathsf{BFV}$ make use of different slot packing methods~\cite{SV14,CKKS17}, which produce some variations in the operations performed before and after encryption and decryption. However, because our interest in this work lies in homomorphic additions, for which SIMD-type operations hold without packing, we have removed this step in both schemes. Consequently, some of the procedures described in Table~\ref{tab:bfv-ckks} become identical for both schemes. 


\begin{table}[!t] 
	\renewcommand{\arraystretch}{1.9}
	\caption{Summary of $\mathsf{BFV}$ \& $\mathsf{CKKS}$ HE Schemes}
	\label{tab:bfv-ckks}
    \scriptsize 
	
	\begin{tabular}{|m{1.3in}|m{0.25in}|m{1.35in}|} 
		\hline
		\multicolumn{3}{|c|}{Parameters remarks}\\
		\hline
		\multicolumn{3}{|p{3.3in}|}{ 
    
        $\mathsf{BFV}$ \& $\mathsf{CKKS}$: $R_q$ is the ciphertext ring.
        \newline $\mathsf{BFV}$: 
                $\Delta = \lfloor \frac{q}{t} \rfloor$. $R_t$ is the plaintext ring, with $q > t$.
                  \newline $\mathsf{CKKS}$: 
                  Given an arbitrary target margin error $\epsilon$ and $B$, 
                  we can set the scale $\Delta$. 
		}\\
		\hline
		
		\multicolumn{3}{|c|}{Cryptographic Primitives}\\
        \hline
		\multirow{2}{*}{$\mathsf{\mathsf{E}_\mathsf{ahe}.Setup()} = pp$}
		& \textsf{BFV} & Select $pp \leftarrow (t, n, q, \sigma, B, \Delta)$ \\
		\cline{2-3}
		& \textsf{CKKS} & Select $pp \leftarrow (\epsilon, 
  n, q, \sigma, B, \Delta)$ \\
		\hline
		$\mathsf{\mathsf{E}_\mathsf{ahe}.SecKeyGen}(1^\lambda) = \mathsf{sk}$
		& \multicolumn{2}{p{1.6in}|}{Given the security parameter $\lambda$, sample $s \leftarrow R_3$. Output $\mathsf{sk} = s$} \\
		\hline
		$\mathsf{\mathsf{E}_\mathsf{ahe}.PubKeyGen(sk)} = \mathsf{pk}$
		&  \multicolumn{2}{p{1.6in}|}{Sample $p_1 \leftarrow R_q$, and $e \leftarrow \chi$. $\mathsf{pk} = (p_0, p_1) = (-\mathsf{sk}\cdot p_1 + e, p_1)$} \\		
		\hline
		$\mathsf{\mathsf{E}_\mathsf{ahe}.Enc}(\mathsf{pk}, m) = \mathsf{ct}$
		& \multicolumn{2}{p{1.6in}|}{Sample $u \leftarrow R_3$ and $e_0, e_1 \leftarrow \chi$. $\mathsf{ct} = (c_0,c_1) = (\Delta m + u \cdot p_0 + e_0, u \cdot p_1 + e_1)$} \\
		\hline
  	$\mathsf{\mathsf{E}_\mathsf{ahe}.Add}(\mathsf{ct}, \mathsf{ct}') = \mathsf{ct}_{\mathsf{add}}$
		&  \multicolumn{2}{p{1.6in}|}{Given $\mathsf{ct} = (c_0, c_1)$ and $\mathsf{ct}' = (c_0', c_1')$, $\mathsf{ct}_{\mathsf{add}} = (c_0 + c_0', c_1 + c_1')$} \\	
		\hline
		\multirow{2}{*}{$\mathsf{\mathsf{E}_\mathsf{ahe}.Dec}(\mathsf{sk}, \mathsf{ct}) = m$}
		& \textsf{BFV} & $m = {[\lfloor \frac{t}{q} {[c_0 + c_1 \cdot s]}_q \rceil]}_t$ \\
		\cline{2-3}
		& \textsf{CKKS} & $m \approx \frac{{[c_0 + c_1\cdot s]}_q}{\Delta}$ \\
		\hline
  \multicolumn{3}{|c|}{Decryption remarks}\\
		\hline
		\multicolumn{3}{|p{3.3in}|}{ 
        $\mathsf{BFV}$: Decryption correctness holds if $(2n +1)B < \frac{q}{2t} - \frac{t}{2}$ (see~Lemma~\ref{lemma:correct-decryption-single-key}), being $n > \delta_R$ an upper bound for the expansion factor $\delta_R$ of the ring $R$~\cite{deCJV21}. 
        \newline $\mathsf{CKKS}$: 
        After decryption, we obtain $\tilde{m} = m + e_{ct} \approx m$, where $e_{ct}$ is the internal error of the ciphertext, with $||e_{ct}|| < \frac{(2n + 1)\cdot B}{\Delta}$. 
		}\\
		\hline
	\end{tabular}
 \vspace{-0.4cm}
\end{table}

\subsection{A Protocol for Private Aggregation with Multiparty HE}
\label{sec:mhe}
To securely evaluate the aggregation functionality $f_\Sigma(\cdot)$, we resort to the extension of the single-key HE schemes from Section~\ref{sec:background} into their threshold variants. In particular, we make use of the Multiparty Homomorphic Encryption (MHE) scheme introduced in~\cite{MTBH21} as a solution to the following \emph{multiparty aggregation problem}: 
\begin{definition}[Adapted from Def.~$1$ in~\cite{MTH19} to our FL setting]\label{def:mpc} Let $\mathcal{C} = \{C_1, C_2, \ldots, C_L\}$ be a set of $L$ clients, where each client $C_i$ holds an input $\vec{w}^i$ (\emph{input and receiver parties}). Let $f_\Sigma(\vec{w}^1, \vec{w}^2, \ldots, \vec{w}^L) = \vec{w}^*$ be the average aggregation function ($f_{\Sigma}$ is called \emph{ideal functionality}) over the \emph{input parties}. Let $\mathcal{A}$ be a static semi-honest adversary that can corrupt up to $L - 1$ clients in $\mathcal{C}$ and let $\mathcal{C_{A}}$ be the subset of clients corrupted by $\mathcal{A}$. Then, the \emph{secure multiparty aggregation problem} consists in providing $\mathcal{C}$ with $\vec{w}^*$, yet $\mathcal{A}$ must learn nothing more about $\{\vec{w}^i\}_{C_i \notin \mathcal{C_{A}} }$ than what can be deduced from the inputs $\{\vec{w}^i\}_{C_i \in \mathcal{C_{A}}}$ and output $\vec{w}^*$ it controls (this property is called \emph{input privacy}).
\end{definition}

Consequently, a solution to the \emph{secure multiparty aggregation} problem consists in a protocol $\pi_{f_\Sigma}$ which realizes the \emph{ideal functionality} $f_\Sigma$ while also preserving the \emph{input privacy}~\cite{Lindell17}.

\begin{table}[!t] 
	\renewcommand{\arraystretch}{1.3}
	\caption{Multiparty Private Aggregation Protocol}
	\label{tab:mhempc}
	\centering \small 
	\begin{threeparttable}[b]
	\begin{tabular}{|m{0.5in}|m{2.6in}|} 
		\hline
		\multicolumn{2}{|p{3.2in}|}{
		
    \textbf{Public input}: ideal aggregation functionality $f_\Sigma$ to be  computed\newline
    \textbf{Private input}: $\vec{w}^i$ for each $C_i \in \mathcal{C}$\newline
    \textbf{Output} for all clients in $\mathcal{C}$: $\vec{w}^* = f(\vec{w}^1, \vec{w}^2, \ldots, \vec{w}^L)$
		}\\
		\hline
		
		Setup & All clients $C_i$ instantiate the multiparty homomorphic scheme $\mathsf{E_{mhe}}$ \newline
		$\mathsf{sk}_i = \mathsf{E_{mhe}}.\pi_{i, \mathsf{SecKeyGen}}(\lambda, \kappa)$\tnote{*} \newline
		$\mathsf{cpk} = \mathsf{E_{mhe}}.\pi_{\mathsf{PubKeyGen}}(\kappa, \mathsf{sk}_1, \ldots, \mathsf{sk}_L)$ \\		
		\hline
		Input & Each $C_i$ encrypts its input $\vec{w}^i$ and provides it to the \textit{Aggregator}\tnote{+} \newline
		$c_{\vec{w}^i} = \mathsf{E_{mhe}.Enc}(\mathsf{cpk}, \vec{w}^i)$\\		
		\hline
		Evaluation & The \textit{Aggregator} computes the encrypted output for the ideal functionality $f_\Sigma$ relying on $\mathsf{E_{mhe}.Add}$\newline
		$c_{\vec{w}^*}= f_\Sigma(c_{\vec{w}^1}, c_{\vec{w}^2}, \ldots, c_{\vec{w}^L})$\\
		\hline
		Output & The parties in $\mathcal{C}$ execute the decryption protocol \newline
		$\vec{w}^* = \mathsf{E_{mhe}}.\pi_{\mathsf{Dec}}(\mathsf{sk}_1, \ldots, \mathsf{sk}_L, c_{\vec{w}^*})$\\
		\hline
	\end{tabular}
	   \begin{tablenotes}
     \item[*] $\kappa$ parameterizes the homomorphic capacity of the scheme $\mathsf{E_{Add}}$.
     \item[+] The \textit{Aggregator} is in charge of performing the aggregation.
   \end{tablenotes}
  \end{threeparttable}
  \vspace{-0.8cm}
\end{table}

We can make use of the general MHE-based solution proposed in~\cite{MTBH21}, which is proved on the Common Reference String (CRS) model,\footnote{It assumes that all parties have access to a common random string.} and it also assumes that all parties are connected through authenticated channels. The description of our particularized protocol for private aggregation is described in Table~\ref{tab:mhempc}. It relies on the existence of the following augmented multiparty encryption scheme: 
\begin{definition}[Multiparty Additive HE]\label{def:mhe} Let $\mathsf{E}_\mathsf{ahe}$ 
be the asymmetric and additive homomorphic encryption scheme from Definition~\ref{def:addhe}, whose security is parameterized by $\lambda$ (see Table~\ref{tab:bfv-ckks} for two concrete examples with $\mathsf{BFV}$ and $\mathsf{CKKS}$), and let $\mathsf{S} = (\mathsf{S.Share}, \mathsf{S.Combine})$ be an $L$-party secret sharing scheme. The associated \emph{multiparty homomorphic encryption scheme} ($\mathsf{E}_{\mathsf{mhe}}$) is obtained by applying the secret-sharing scheme $\mathsf{S}$ to $\mathsf{E_{ahe}}$'s secret key $\mathsf{sk}$ (\emph{ideal secret key}) and is defined as the tuple $\mathsf{E}_{\mathsf{mhe}} = (\mathsf{E_{ahe}.Enc}, \mathsf{E_{ahe}.Dec}, \mathsf{E_{ahe}.Add}, \mathsf{E_{ahe}}^{\mathsf{S}})$, where $\mathsf{E_{ahe}}^{\mathsf{S}} = (\pi_{\mathsf{SecKeyGen}}, \pi_{\mathsf{PubKeyGen}}, \pi_{\mathsf{Dec}})$ is a set of multiparty protocols executed among the clients in the set $\mathcal{C}$, and having the following \emph{private ideal functionalities} for each client $C_i$:
\begin{itemize}
    \item \emph{Ideal secret-key generation}:
    
    $f_{i,\pi_{\mathsf{SecKeyGen}}}(\lambda) = \mathsf{S}.\mathsf{Share}_i(\mathsf{E_{ahe}}.\mathsf{SecKeyGen}(\lambda)) = \mathsf{sk}_i$.
    \item \emph{Collective public-key generation}:
    
    $f_{\pi_{\mathsf{PubKeyGen}}}(\mathsf{sk}_1, \mathsf{sk}_2, \ldots, \mathsf{sk}_L)$ \\ $= \mathsf{E_{ahe}}.\mathsf{PubKeyGen}(\mathsf{S}.\mathsf{Combine}(\mathsf{sk}_1, \mathsf{sk}_2, \ldots, \mathsf{sk}_L))$.
    \item \emph{Collective decryption}:

    $f_{\pi_{\mathsf{Dec}}}(\mathsf{sk}_1, \mathsf{sk}_2, \ldots, \mathsf{sk}_L, \mathsf{ct})$ \\ $= \mathsf{E_{ahe}}.\mathsf{Dec}(\mathsf{S}.\mathsf{Combine}(\mathsf{sk}_1, \mathsf{sk}_2, \ldots, \mathsf{sk}_L),\mathsf{ct})$.
\end{itemize}

\end{definition}
Definition~\ref{def:mhe} particularizes Def. $2$ from~\cite{MTH19} to the case in which the encryption scheme $\mathsf{E_{ahe}}$ is additive homomorphic. 

\textbf{Concrete instantiations for Multiparty HE:}
The private ideal functionalities introduced in Definition~\ref{def:mhe} can be implemented by the clients with concrete protocols: 
\begin{itemize}
    \item $f_{i,\pi_{\mathsf{SecKeyGen}}}(\lambda)$: Each client $C_i$ independently runs the procedure $\mathsf{\mathsf{E}_\mathsf{Add}.SecKeyGen}(1^\lambda) = \mathsf{sk}_i$.
    \item $f_{\pi_{\mathsf{PubKeyGen}}}(\mathsf{sk}_1, \mathsf{sk}_2, \ldots, \mathsf{sk}_L)$: Given a common random polynomial $p_1$, each client $C_i \in \mathcal{C}$ samples $e_{\pk,i} \leftarrow \chi$ and discloses to the other clients $p_{0,i} = -p_1 \cdot \mathsf{sk_i} + e_{\pk,i}$. The collective public key is computed as: 
    \begin{equation*}\mathsf{cpk} = (\displaystyle \sum_{C_i \in \mathcal{C}}p_{0,i},p_1) = (-p_1\underbrace{\sum_i \mathsf{sk}_i}_{\mathsf{sk}} + \sum_i e_{\pk,i},p_1).
    \end{equation*}
    \item $f_{\pi_{\mathsf{Dec}}}(\mathsf{sk}_1, \mathsf{sk}_2, \ldots, \mathsf{sk}_L, \mathsf{ct})$: The corresponding collaborative decryption protocol can be divided in two phases.
    \begin{itemize}
        \item Given a ciphertext $\mathsf{ct}  = (c_0, c_1)$ encrypted under the ideal secret-key $\mathsf{sk}$, each client computes its partial decryption of $\mathsf{ct}$.  Each client $C_i$ samples $e_{\textsf{smg},i} \leftarrow \chi$ and discloses:
            \begin{equation*}
                h_i = \mathsf{sk_i}\cdot c_1 + e_{\textsf{smg},i}.
            \end{equation*}
        \item Given all the decryption shares $h_i$ from the clients, the protocol outputs:
            \begin{align*}
             d & \mbox{ } = \mbox{ } {\left[ c_0 + \sum_{C_i \in \mathcal{C}}h_{i}\right]}_q \\ & \mbox{ } = \mbox{ } \Delta m + \underbrace{e\cdot v + e_0 + s\cdot e_1}_{e_\mathsf{ct}} + \underbrace{\sum_i e_{\mathsf{smg},i}}_{e_{\mathsf{smg}}}.
            \end{align*}
    \end{itemize}
\end{itemize}

The final step for decryption depends on whether we use $\mathsf{BFV}$ or $\mathsf{CKKS}$ as the baseline single-key HE scheme:
\begin{itemize}
    \item For $\mathsf{BFV}$, the obtained decryption is $m = {[ \lfloor \frac{t}{q}d \rceil ]}_t$.
    \item For $\mathsf{CKKS}$, the obtained decryption is $m + \frac{e_\mathsf{ct} + e_{\mathsf{smg}}}{\Delta} = \frac{d}{\Delta}$.
\end{itemize}
We can obtain lower bounds for $q$, similar to the single-key counterparts, by taking into account the larger noise terms present in the threshold versions (see Section~\ref{sec:comp-bfv-ckks}).

\subsection{Multiparty HE with Countermeasures for $\mathsf{IND}$-$\mathsf{CPA}^\mathsf{D}$}
\label{sec:countermeasures}
The new notion $\mathsf{IND}$-$\mathsf{CPA}^\mathsf{D}$ extends $\mathsf{IND}$-$\mathsf{CPA}$ by allowing the adversary to have access to a very restricted decryption oracle which can only be used for genuine ciphertexts, or ciphertexts obtained from genuine ciphertexts through valid homomorphic operations. By adapting the definition to our specific case with the aggregation function $f_\Sigma$, the idea is that an adversary knowing $\{\vec{w}^1, \ldots, \vec{w}^L \}$ and $f_\Sigma$, can also obtain $f_{\Sigma}(\vec{w}^1,\ldots,\vec{w}^L)$. Then, if this adversary has access to $\{\mathsf{Enc}(\vec{w}^1), \ldots, \mathsf{Enc}(\vec{w}^L)\}$ and homomorphically evaluates the aggregation function $f_\Sigma$ as $f_\Sigma(\mathsf{Enc}(\vec{w}^1), \ldots, \mathsf{Enc}(\vec{w}^L))$, then she should not gain more information from $\mathsf{Dec}\left(f_\Sigma \left( \mathsf{Enc}(\vec{w}^1), \ldots, \mathsf{Enc}(\vec{w}^L) \right)\right)$ than what she can already obtain from $f_{\Sigma}(\vec{w}^1,\ldots,\vec{w}^L)$. Unfortunately, in~\cite{LM21} the authors showed how some HE schemes, which are $\mathsf{IND}$-$\mathsf{CPA}$ secure, leak information through its difference 
and, hence, are not $\mathsf{IND}$-$\mathsf{CPA}^\mathsf{D}$ secure. An example is the case of approximate HE schemes like $\mathsf{CKKS}$~\cite{CKKS17}, for which 
the difference $\mathsf{Dec}(\mathsf{Enc}(\vec{w}^i)) - \vec{w}^i$ directly leaks the internal noise of the underlying RLWE sample, then allowing the adversary to break the RLWE indistinguishability assumption.

Till very recently, the cryptographic community believed that only approximate HE presented this vulnerability and that exact HE schemes like $\mathsf{BFV}$~\cite{Brakerski12,FV12}, $\mathsf{BGV}$~\cite{BGV14} or $\mathsf{CGGI}$~\cite{CGGI16} were invulnerable to this type of attacks, hence being these schemes naturally safe under $\mathsf{IND}$-$\mathsf{CPA}^\mathsf{D}$. However, a couple of new works~\cite{CSBBC24,CCPSS24} have shown that this vulnerability could also be present in practice for exact schemes, as soon as they present a non-negligible probability of incorrect decryption. Both works exemplify how this vulnerability can be exploited to implement effective key-recovery attacks against several mainstream HE libraries.

In both works~\cite{CSBBC24,CCPSS24}, the authors illustrate how these types of attacks are especially relevant when dealing with threshold HE schemes. Fortunately, they also propose several countermeasures to achieve $\mathsf{IND}$-$\mathsf{CPA}^\mathsf{D}$ security when using threshold HE. In particular, the addition during decryption of an adequate smudging noise with $\lambda$-independent variance can do the job. It is worth highlighting that~\cite{CSBBC24} showcases how this countermeasure significantly impacts the cryptosystem's parameters, subsequently reducing its efficiency. According to the authors, this effect is even more pronounced for $\mathsf{CKKS}$, for which using a large-variance smudging noise is likely to severely reduce the precision of the decrypted result.


\subsubsection{Impact of the smudging noise}


In the end, given a ciphertext $\mathsf{ct}$, the key-recovery attacks mentioned above aim at extracting its noise component $e_\mathsf{ct}$. If this extraction is feasible, the secret key can be recovered via linear algebra techniques. For approximate HE schemes, this process is relatively easy once we have access to the decryption. However, for exact HE schemes, these attacks attempt to force decryption failures as a mechanism to estimate the corresponding $e_\mathsf{ct}$ component.

Therefore, even in extreme cases where the post-processed decryption can become unusable, adding a large enough smudging noise after decryption is required to hide any information leakage regarding the original error $e_\mathsf{ct}$.

In relation to this countermeasure, recent work by~\cite{ABMP24} discusses how the \textit{application-agnostic} nature of $\mathsf{IND}$-$\mathsf{CPA}^\mathsf{D}$ often leads to impractically large parameters when adding a large smudging noise. This is due to the fact that $e_\mathsf{smg}$ must statistically hide $e_\mathsf{ct}$ for all possible homomorphic circuits which satisfy correctness for the initially chosen cryptosystem parameters with the $\mathsf{E_{ahe}}.\mathsf{Setup}()$ procedure. The authors propose a relaxation of this definition, termed \textit{application-aware} $\mathsf{IND}$-$\mathsf{CPA}^\mathsf{D}$, which is more suitable for the FL setting.

In our work, we follow their guidelines and assume that both Clients and the Aggregator will behave semi-honestly
, which here means that they are \textit{expected to follow exactly the prescribed instructions to compute} $c_{\vec{w}^*}= f_\Sigma(c_{\vec{w}^1}, c_{\vec{w}^2}, \ldots, c_{\vec{w}^L})$ \textit{and encrypt} $\{c_{\vec{w}^1}, c_{\vec{w}^2}, \ldots, c_{\vec{w}^L}\}$ (see Tables~\ref{tab:bfv-ckks} and~\ref{tab:mhempc}).

\textbf{Collaborative decryption protocol:} Given again $\mathsf{ct}$, once all partial decryptions $h_i = \mathsf{sk_i}\cdot c_1 + e_{\textsf{smg},i}$ are gathered, the adversary has direct access to $d = \Delta \cdot m + e_{\mathsf{ct}} + e_{\textsf{smg}}$. Our objective is to obtain $e_{\textsf{smg}}$ such that $e_\mathsf{ct} + e_{\textsf{smg}}$ is statistically indistinguishable from fresh noise $\tilde{e}_{\textsf{smg}}$ (see the Smudging Lemma in~\cite{AJLTVW12}). To ensure $\mathsf{IND}-\mathsf{CPA}^\mathsf{D}$ security, we adhere to the practical guidelines provided in~\cite{CSBBC24}, which recommend a Gaussian smudging noise with variance $\sigma_{\mathsf{smg}}^2 = 2^\lambda \sigma_{\mathsf{ct}}^2$. 
By applying an upper-bound $B = \mathcal{O}(\sigma)$,\footnote{$B = 6\sigma$ is commonly used by HE libraries.}
 this results in $B_{\mathsf{smg}} = 2^{\frac{\lambda}{2}} B_{\mathsf{ct}}$ (e.g., $\lambda = 128$ is suggested in~\cite{MTBH21}).

\subsection{Comparison analysis between $\mathsf{BFV}$ and $\mathsf{CKKS}$}
\label{sec:comp-bfv-ckks}

In~\cite{CSBBC24}, the authors discuss how the fact that $\mathsf{CKKS}$ is particularly vulnerable to the exposure of decryptions under $\mathsf{IND}$-$\mathsf{CPA}^\mathsf{D}$ makes working with its threshold variant, by following the blueprint of~\cite{MTBH21} (see Section~\ref{sec:mhe}), less useful compared to other exact schemes such as $\mathsf{BFV}$. The main argument provided relies on the fact that the high variance of the required smudging noise is likely to jeopardize the precision of the finally decrypted results. Our aim in this section is to elaborate more precisely on this comparison in the context of FL. In particular, we compare the bit precision achieved by the threshold variants of both $\mathsf{BFV}$ and $\mathsf{CKKS}$ for the homomorphic evaluation of the average aggregation rule.

\textbf{Choosing multiparty cryptosystems' parameters:} We can upper-bound the noise of a fresh single-key ciphertext by $(2n + 1)B$, relying on Lemma~\ref{lemma:correct-decryption-single-key} (see Appendix). To extend this result to the multiparty HE schemes from Section~\ref{sec:mhe}, we must consider the following: (1) The secret key and error for the $\mathsf{cpk}$ are $L$ times larger, which implies that the noise of fresh ciphertexts is now upper-bounded by $B(2nL + 1)$. (2) During the homomorphic execution of the $f_\Sigma$ functionality, the underlying noise is further increased by a factor of $L$. (3) During the collective decryption of the resulting $\mathsf{ct}$, every client adds to $e_\mathsf{ct}$ a smudging noise satisfying $||e_{\mathsf{smg},i}|| \leq B_{\mathsf{smg}}$.

As a result, the final noise $e_\mathsf{ct} + e_\mathsf{smg}$ present in the collective decryption of $\mathsf{ct}$ is upper-bounded by $B^{\mathsf{MP}}_\mathsf{ct} = B_\mathsf{ct} + LB_\mathsf{smg} =  (1 + L2^{\lambda/2})B_\mathsf{ct}$, with $B_\mathsf{ct} = LB(2nL + 1)$. As $q = \mathcal{O}(B^\mathsf{MP}_\mathsf{ct})$ (see Prop.~\ref{prop:qmhe}) for both $\mathsf{BFV}$ and $\mathsf{CKKS}$, we can see how the smudging noise significantly impacts ciphertext size and, subsequently, the efficiency of the private aggregation protocol. 

More specifically, the remarks for decryption with single-key HE given in Table~\ref{tab:bfv-ckks} can be adapted to their threshold variants by considering $B^{\mathsf{MP}}_\mathsf{ct}$ instead of $B_\mathsf{ct}$. This results in the expressions: (1) $B^{\mathsf{MP}}_{\mathsf{ct}} <  \frac{q}{2t} - \frac{t}{2}$ for Multiparty $\mathsf{BFV}$ ($\mathsf{MBFV}$) and, (2) $\Delta B_m + B^{\mathsf{MP}}_{ct} < q/2$, with $||m|| < B_m$ and $m = \sum_i^L m_i$, for Multiparty $\mathsf{CKKS}$ ($\mathsf{MCKKS}$).

\textbf{A comparison in terms of cipher expansion:} From the previous expressions, we can obtain lower-bounds for $q$, which, by following Prop.~\ref{prop:qmhe} (see Appendix), can be used to compare the cipher expansion of both schemes. In particular, $\mathsf{MCKKS}$ has a smaller $q$ than $\mathsf{MBFV}$ if:
\begin{equation} \label{EC}
    B^{\mathsf{MP}}_{\mathsf{ct}} > \frac{2\Delta_{\mathsf{CKKS}} B_m - t^2}{2(t - 1)},
\end{equation}
where we use $\Delta_{\mathsf{CKKS}}$ to make explicit that we refer to the $\Delta$ from $\mathsf{CKKS}$ in Table~\ref{tab:bfv-ckks}. Note that $\Delta_{\mathsf{CKKS}}$ and $t$ do not play the same role in $\mathsf{MCKKS}$ and $\mathsf{MBFV}$, respectively.

To ensure a fair comparison between both schemes, we express $\Delta_\mathsf{CKKS}$ in terms of the error margin $\epsilon = B_\mathsf{ct}^\mathsf{MP}/\Delta_\mathsf{CKKS}$ introduced in Section~\ref{sec:background}. This allows us to compare the achieved bit precision of $\mathsf{MCKKS}$, represented by $\log_2{(\epsilon^{-1}B_m)}$, directly with the bit precision $\log_2{t}$ achieved in $\mathsf{MBFV}$. Also, for simplicity in the comparison, we assume that the input plaintexts used in $\mathsf{MCKKS}$ are normalized beforehand to ensure an aggregation satisfying $B_m < 1$. These changes lead to the following inequality:
\begin{equation} \label{eq:he-comparison}
    \frac{t^2}{2B^{\mathsf{MP}}_{\mathsf{ct}}} + t - 1 > \epsilon^{-1}. 
\end{equation}

In Figures~\ref{BFV against CKKS} and~\ref{fig:varyingL}, we plot the expression from Eq.~\eqref{eq:he-comparison} on a logarithmic scale, using bit precisions $\log_2{t}$ and $\log_2{\epsilon^{-1}}$. Note that if a point ($\log_2{t}, \log_2{\epsilon^{-1}}$) belongs to the colored area, then the modulus $q$ required for $\mathsf{MCKKS}$ is smaller than that of $\mathsf{MBFV}$. In Figure~\ref{BFV against CKKS}, we set parameters to $L = 10$, $n = 8192$ and $B = 19.2$, while $\lambda$ ranges over the set $\{32, 64, 96, 128\}$. In Figure~\ref{fig:varyingL}, we fix the parameters at $n = 8192$, $B = 19.2$, and $\lambda = 128$, while $L$ takes values in $\{8, 128\}$. The parameter values are practical enough for the cross-silo setting, and ensure a bit security of at least $128$ for all the represented bit precisions. 

We observe that while increasing $\lambda$ and $L$ up to $128$ does reduce the range of parameters for which $\mathsf{MCKKS}$ outperforms $\mathsf{MBFV}$ (the effect is more significant when increasing $\lambda$), the approximate $\mathsf{MCKKS}$ compares similarly with $\mathsf{MBFV}$, getting better for very high bit precisions. Finally, by taking a closer look to the figures, we can also see how \eqref{eq:he-comparison} approximately follows a piecewise linear function with two intervals:
\begin{itemize}
\item In the first interval, the addend $t-1$ dominates the left part of Eq.~\eqref{eq:he-comparison}, allowing us to simplify the inequality to approximately $\log_2({t - 1}) > \log_2{\epsilon^{-1}}$. Here, both schemes behave similarly for the same bit precision.
\item In the second interval, the addend $\frac{t^2}{2B_\mathsf{ct}^\mathsf{MP}}$ dominates, simplifying the inequality into approximately $2\log_2{t} - \log_2{B_\mathsf{ct}^\mathsf{MP}} - 1 > \log_2{\epsilon^{-1}}$. In this case, the bit precision of $\mathsf{MCKKS}$ grows twice as fast as $\mathsf{MBFV}$. 
\end{itemize}

\begin{figure}[h!]
\vspace{-0.2cm}
\centering
\subfloat[$\lambda = 32$]{\includegraphics[width = 1.6in]{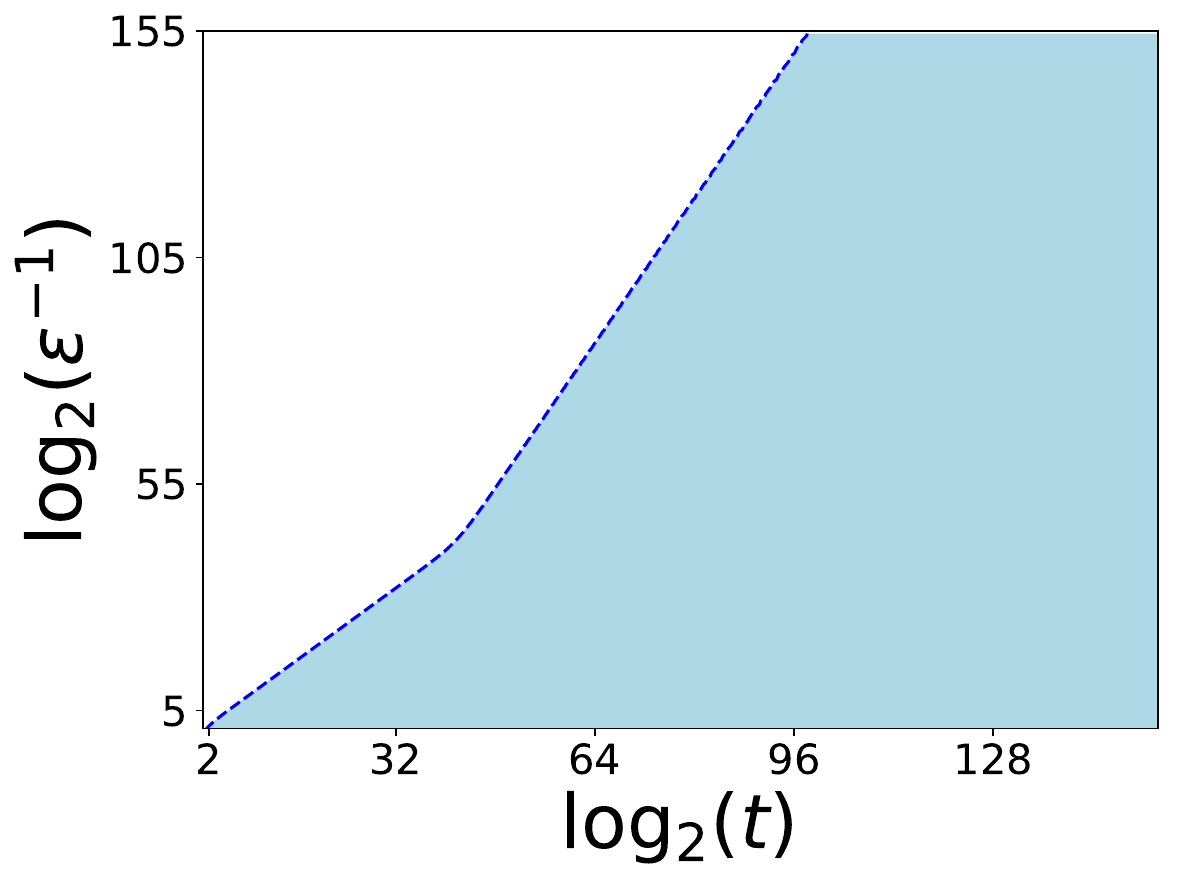}} \hspace{.2in}
\subfloat[$\lambda = 64$]{\includegraphics[width = 1.6in]{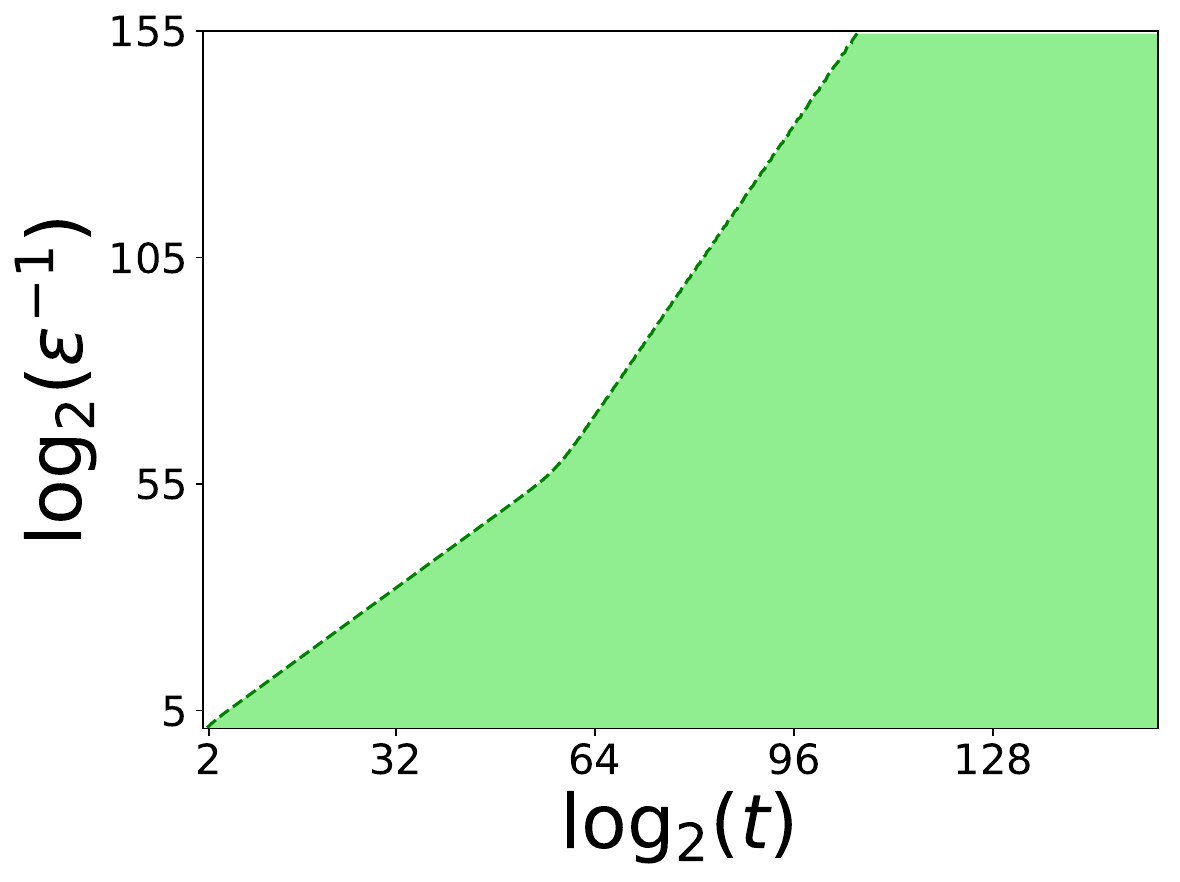}}\\
\subfloat[$\lambda = 96$]{\includegraphics[width = 1.6in]{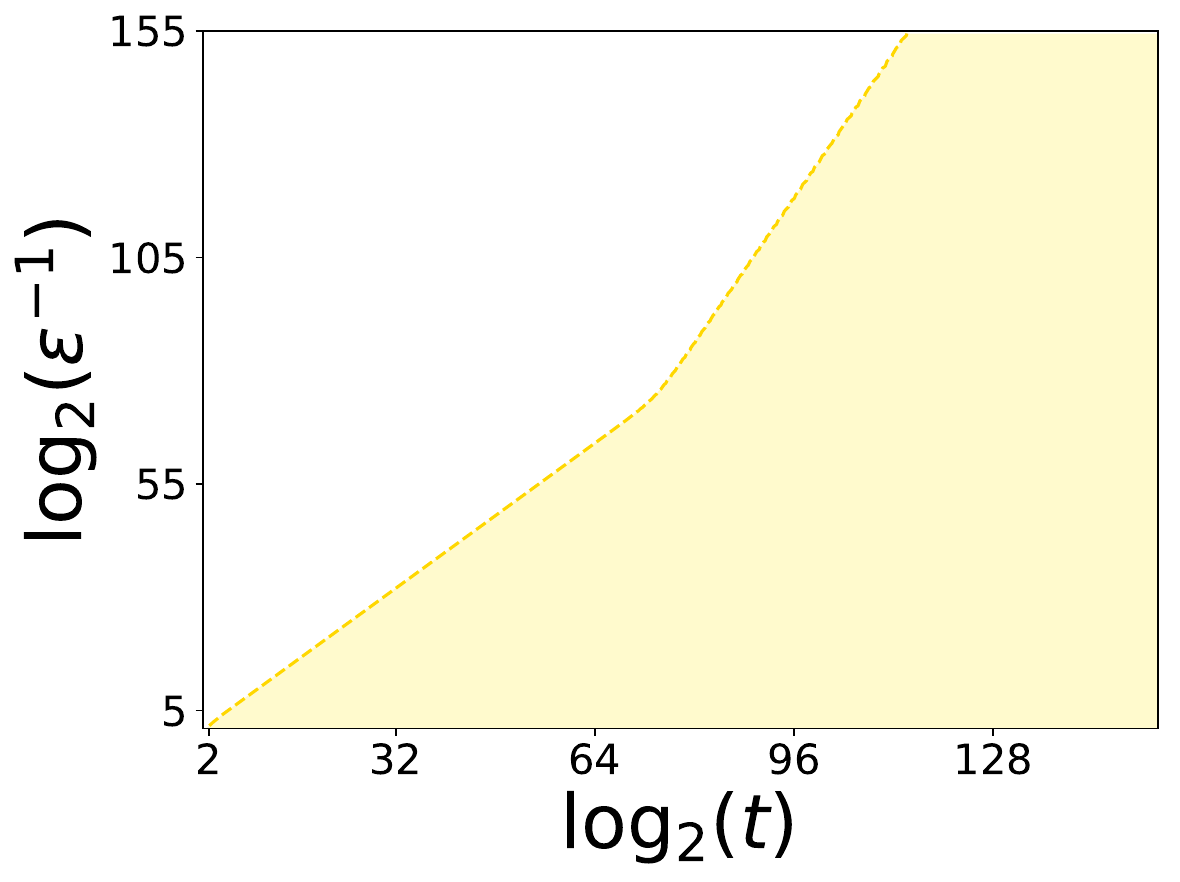}}\hspace{.2in}
\subfloat[$\lambda = 128$]{\includegraphics[width = 1.6in]{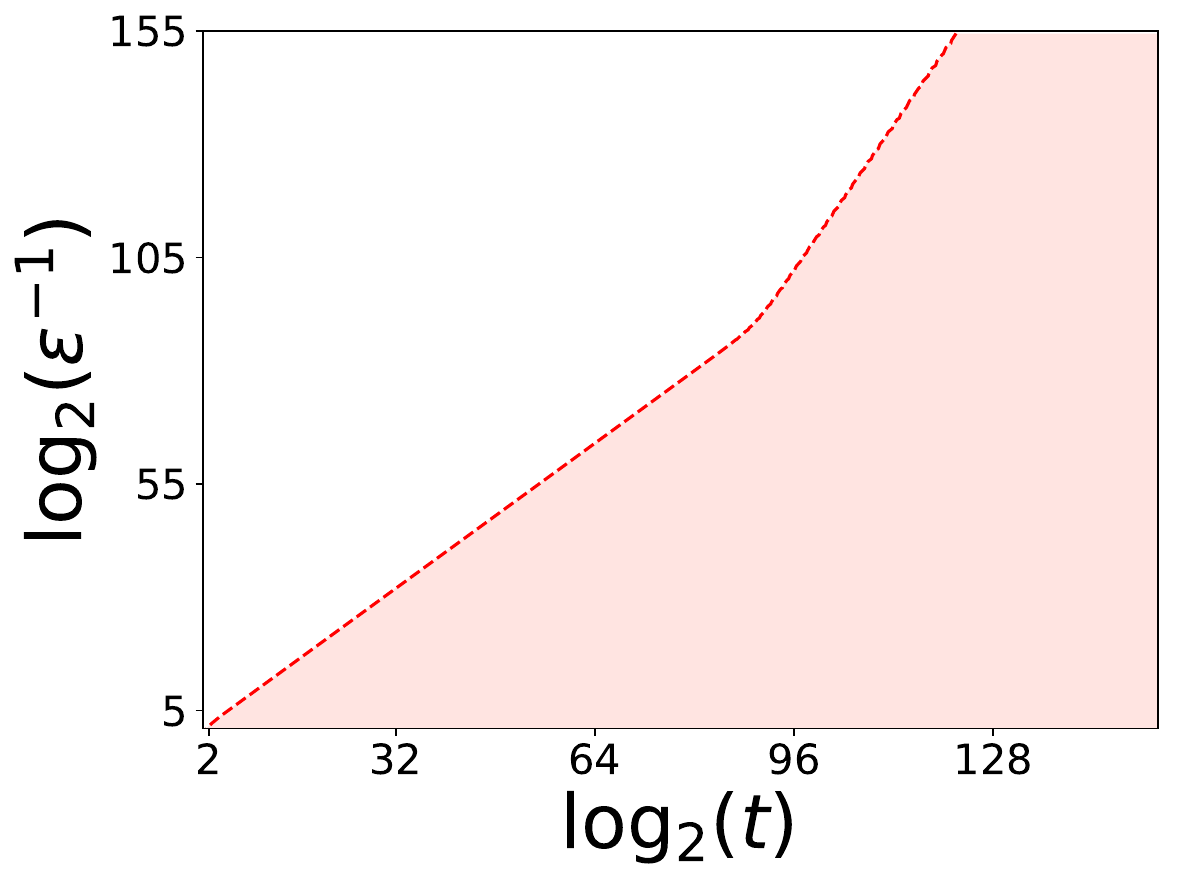}}
 
\caption{Comparison of $q$ for $\mathsf{MBFV}$ and $\mathsf{MCKKS}$ varying $\lambda$.}
\label{BFV against CKKS}
\vspace{-0.1cm}
\end{figure}

\begin{figure}[h!]
\centering
\subfloat[$L = 8$]{\includegraphics[width = 1.6in]{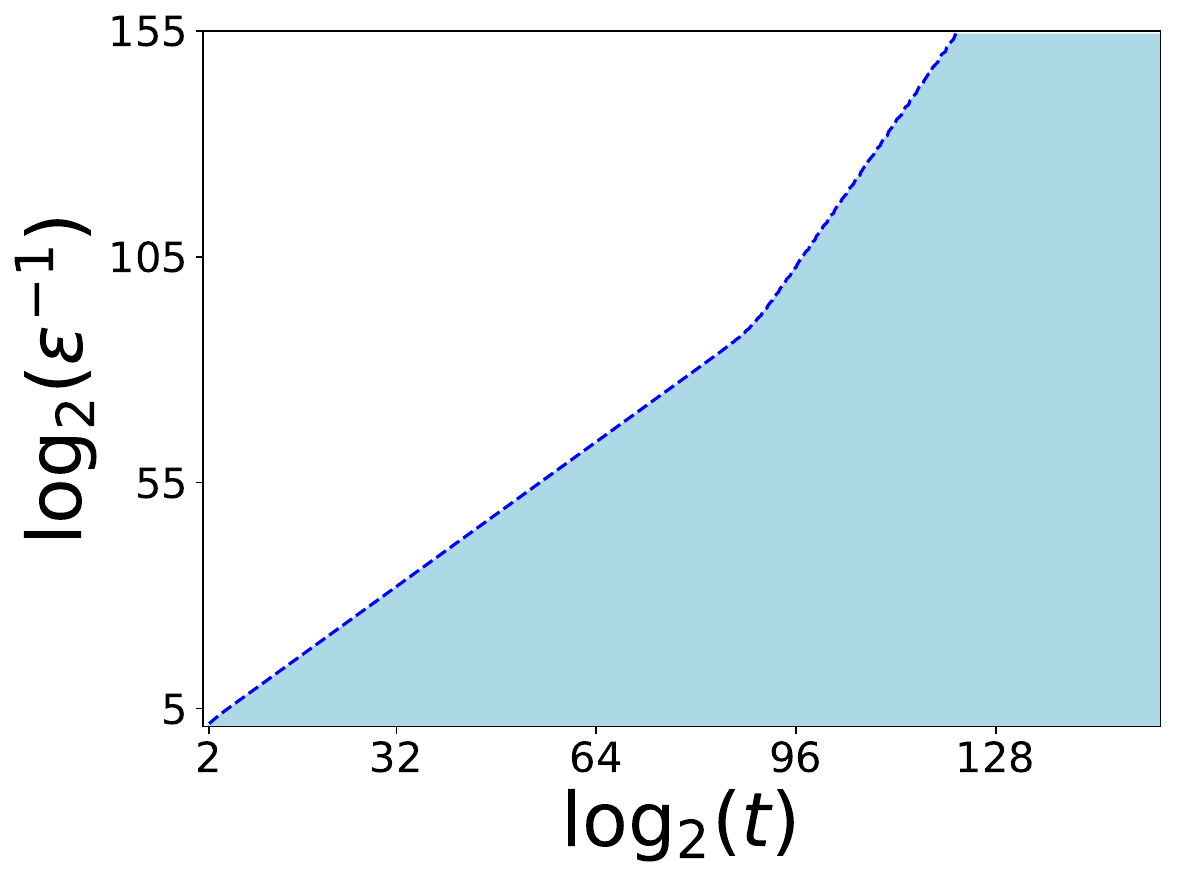}} \hspace{.2in}
\subfloat[$L = 128$]{\includegraphics[width = 1.6in]{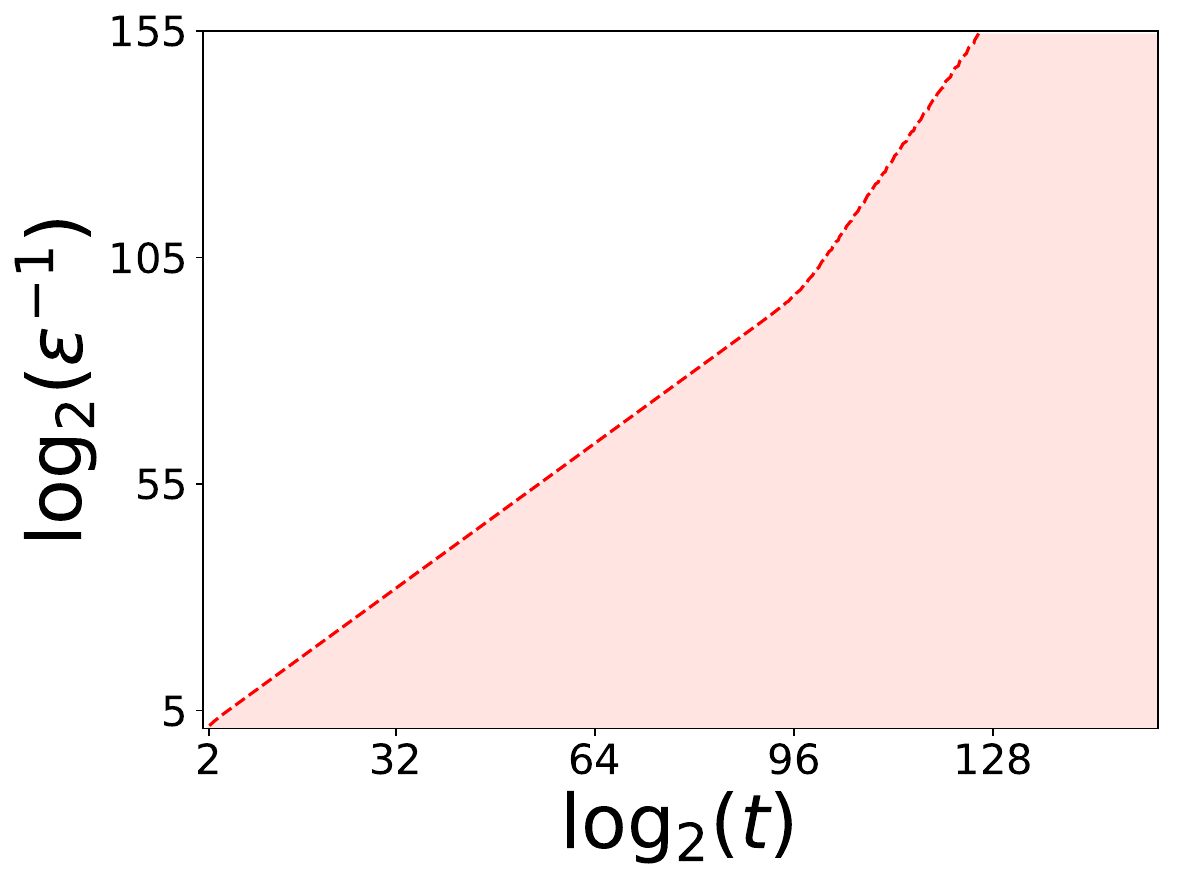}}
 
\caption{Comparison of $q$ for $\mathsf{MBFV}$ and $\mathsf{MCKKS}$ varying $L$.}
\label{fig:varyingL}
\vspace{-0.4cm}
\end{figure}

\textbf{Runtime comparisons:}
The implementation runtimes presented in this section were conducted single-threaded on an Intel Core i7-10750H CPU @ 2.60GHz $\times$ 12 with 31.1 GB. We considered two aggregation protocols with bit precision $\log_2{t}$ and $\log_2{\epsilon^{-1}}$ for $\mathsf{MBFV}$ and $\mathsf{MCKKS}$, respectively. Table~\ref{tab:exampleparams-comparison-mkhe-uvigo} includes three different parameter sets, and Table~\ref{tab:runtimes-comparison-mkhe-uvigo} presents their corresponding implementation runtimes in Lattigo~\cite{MTBH21}.

	\begin{table}[!htbp]
		\centering
			\renewcommand{\arraystretch}{1.3}
			\centering \scriptsize
			\begin{tabular}{c|c|c|c} \hline \hline
				\textbf{Param.} &\textbf{ Set $1$} $\mathsf{BFV}$\&$\mathsf{CKKS}$ & \textbf{Set $2$ $\mathsf{BFV}$} & \textbf{Set $3$ $\mathsf{CKKS}$} \\
				\hline \hline 
    			$\{n, L\}$          & $\{16384, 16\}$                 & $\{16384, 32\}$ & $\{16384, 32\}$ \\
				$\{t, \epsilon^{-1}\}$ \scriptsize{[bits]}          & $\{45, 45\}$ & $\{ 60, -\}$ & $\{ - , 60\}$  \\
                $\{\#\mbox{Limbs}, q$ \scriptsize{[bits]}$\}$          & $\{4, 240\}$ & $\{10, 300\}$ & $\{9, 270\}$  \\
                $\{q_\mathsf{BFV}, q_\mathsf{CKKS}\}$ \scriptsize{[bits]}          & $\{232, 238\}$ & $\{ 280, -\}$ & $\{ - , 259\}$  \\
				\hline \hline 
			\end{tabular}
		\caption{Example parameter sets for the private aggregation protocol ($\sigma = 3.2$, $\lambda = 128$ and $\mathsf{bit} \mbox{ } \mathsf{security} = 128$).}
        \label{tab:exampleparams-comparison-mkhe-uvigo}
  \vspace{-0.4cm}
	\end{table}

	\begin{table}[!htbp]
		\centering
			\renewcommand{\arraystretch}{1.3}
			\centering \scriptsize
			\begin{tabular}{c|c|c|c} \hline \hline 
				\textbf{Protocol step} &\textbf{ Par. set $1$} & \textbf{Par. set $2$} & \textbf{Par. set $3$} \\
				\hline \hline
                Col. Key Gen. & $\mathsf{Client}$: 4.4 $m s$ & $\mathsf{Client}$: 10.3 $m s$ & $\mathsf{Client}$: 10.1 $m s$ \\	
    			Encryption & $\mathsf{Client}$: 1.1 $s$ & $\mathsf{Client}$: 2.3 $s$ & $\mathsf{Client}$: 2.1 $s$  \\
				Aggregation  & $\mathsf{Agg}$: 244.3 $ms$ & $\mathsf{Agg}$: 1.2 $s$ & $\mathsf{Agg}$: 1.1 $s$  \\
                Col. Dec. & $\mathsf{Client}$: 2.5 $s$ & $\mathsf{Client}$: 5.2 $s$ & $\mathsf{Client}$: 4.8 $s$ \\
                Total runtime & 3.9 $s$ & 8.7 $s$ & 8 $s$\\
				\hline \hline
			\end{tabular}
		\caption{Implementation runtimes for the private aggregation protocol ($N_\mathsf{ModelParameters} = 1638400$).}
        \label{tab:runtimes-comparison-mkhe-uvigo}
  \vspace{-0.1cm}
	\end{table}

\textbf{Further clarifications on approximate HE:} There are a few points that we must clarify regarding our comparison. Firstly, we have considered a multiplication by $t/q$ during decryption, following the $\mathsf{BFV}$ description as originally reported in~\cite{FV12}. The rounding error introduced when $t$ does not divide $q$ is precisely the source of the extra addend $\frac{t^2}{2B_\mathsf{ct}^\mathsf{MP}}$ in Eq.\eqref{eq:he-comparison}. If we instead modify the decryption to compute a division by $\Delta$, Eq.\eqref{eq:he-comparison} would resemble the left interval of Figures~\ref{BFV against CKKS} and~\ref{fig:varyingL}.

Consequently, $\epsilon^{-1}$ in $\mathsf{MCKKS}$ grows similarly to $\log_2{t}$ for this modified $\mathsf{MBFV}$. Still, $\mathsf{CKKS}$ presents a significant advantage when used for private aggregation because it naturally introduces noise into the decrypted result. In contrast, exact HE protects the privacy of clients' inputs during homomorphic computation, but does not provide any additional privacy guarantees on the decrypted aggregated models against inference attacks~\cite{MOJC23}. In fact, one possibility could be to incorporate DP along with the exact HE computation~\cite{SCSSG23}. This can be done by homomorphically adding noise into the encryptions, which requires saving part of the $\log_2{t}$ available bits. Instead, following our previous analysis, the bit precision $\epsilon^{-1}$ considered in $\mathsf{CKKS}$ does not suffer any change as the effect of the noise is already contemplated in its definition.

The question of whether the noise present in $\mathsf{CKKS}$ could be reused to provide theoretical DP guarantees has been very recently explored in~\cite{Ogilvie24}. The authors are able to provide a certain privacy budget for the harder case in which the homomorphic evaluation makes the internal noise dependent on the input messages. Note that for the average aggregation rule used in this work, the situation is simpler as the noise added with $\mathsf{CKKS}$ is independent of the input local updates.

\section{Conclusions and future work}
\label{sec:conc}

This work surveys the use of threshold variants of exact and approximate RLWE-based HE for the efficient implementation of private average aggregation. While these types of schemes seem to be a perfect fit for executing the aggregation primitive in FL protocols, recent work has demonstrated the existence of some unexpected security vulnerabilities under the $\mathsf{IND}$-$\mathsf{CPA}^\mathsf{D}$ security model if these HE schemes are not correctly instantiated. We detail the use of smudging noise with a large variance as the main defense proposed by recent works and discuss its impact on performance. We provide an exhaustive analysis comparing the effective bit precision that can be achieved when applying this defense for two concrete threshold variants of HE schemes, such as $\mathsf{BFV}$ and $\mathsf{CKKS}$. Our analysis indicates that $\mathsf{CKKS}$-based average aggregations compare well with $\mathsf{BFV}$-based solutions. Moreover, we discuss how the approximate nature of $\mathsf{CKKS}$ can provide additional privacy guarantees for the decrypted aggregated models against inference attacks. As future work, we intend to explore the use of more complex aggregation rules, optimize the variance required for the smudging noise by taking into account the peculiarities of the FL scenario compared to the more general case of $\mathsf{IND}$-$\mathsf{CPA}^\mathsf{D}$, and provide a more exhaustive performance comparison between both schemes by deploying them within a practical FL use case.

\section*{Acknowledgment}
GPSC is partially supported by the European Union's Horizon Europe Framework Programme for Research and Innovation Action under project TRUMPET (proj. no. 101070038), by FEDER and Xunta de Galicia under project ``Grupos de Referencia Competitiva'' (ED431C 2021/47), by FEDER and MCIN/AEI under project FELDSPAR (TED2021-130624B-C21), and by the ``NextGenerationEU/PRTR'' under TRUFFLES and a Margarita Salas grant of the Universidade de Vigo.

\bibliographystyle{IEEEtran}
\bibliography{biblio.bib}

\appendix
\begin{lemma} \label{lemma:correct-decryption-single-key}
    We follow the notation for $\mathsf{BFV}$ indicated in Table~\ref{tab:bfv-ckks}, and assume that any $e \leftarrow \chi$ satisfies $||e|| \leq B$. For a fresh ciphertext $\mathsf{ct} = (c_0, c_1)$, we have $[c_0 + c_1 \cdot s]_q = \Delta m + e_{\mathsf{ct}}$ with $||e_{\mathsf{ct}}||$ $\leq$ $(2n+1)B$. This implies that whenever $(2n+1)B< \frac{q}{2t} - \frac{t}{2},$ decryption works correctly.
\end{lemma}
\begin{proof}
    We start by computing $[c_0 + c_1 \cdot s]_q = \Delta m + \underbrace{u\cdot e + e_0 + e_1\cdot s}_{e_\mathsf{ct}}$, from which we can directly upper-bound the error polynomial $e_\mathsf{ct}$ as:
    \begin{align}
    ||e_{\mathsf{ct}}|| = & \mbox{ } ||u\cdot e + e_0 + e_1\cdot s|| \\
            \leq & \mbox{ } ||u\cdot e|| + ||e_0|| + ||e_1\cdot s|| \\
            \leq &  \mbox{ } \delta_R B + B + \delta_R B \\
            \leq &  \mbox{ } n B + B + n B \\
            = &  \mbox{ } (2n+1)B.
            \label{eq:bound-error-bfv}
    \end{align}
    Now, the decryption process requires to multiply by $\frac{t}{q}$ and apply a final coefficient-wise rounding:
    $$\left \lfloor\frac{t}{q} ([c_0 + c_1 \cdot s]_q) \right \rceil = \left \lfloor\frac{t}{q}(\Delta m + e_{ct}) \right \rceil,$$
    which must output $m$ for correct decryption. If we define $\Delta = \frac{q}{t} - r$ with $0\leq r < 1$, the condition for correct decryption can be expressed equivalently as: 
    $$||\frac{t}{q}(-rm + e_{ct})|| < \frac 1 2.$$
    By upper-bounding the expression in the left, we obtain:
    \begin{equation}
    \label{eq:bfv-cond-correctness}
    ||\frac{t}{q}(-rm + e_{ct})|| \leq \frac t q(||-rm|| + ||e_{ct}||) < \frac t q(\frac{t}{2}+||e_{ct}||) < 1/2.
    \end{equation}
    Therefore, combining the expressions~\eqref{eq:bound-error-bfv} and~\eqref{eq:bfv-cond-correctness}, we see that decryption correctness holds if $(2n + 1) B < \frac q {2t} -\frac{t}{2}.$
\end{proof}
\begin{prop} \label{prop:qmhe}
    Following the descriptions of the primitives from Tables~\ref{tab:bfv-ckks} and~\ref{tab:mhempc}, the ciphertext modulus $q$ for $\mathsf{MBFV}$ is larger than the $q$ for $\mathsf{MCKKS}$ if the next condition is satisfied: $B^{\mathsf{MP}}_{\mathsf{ct}}>\frac{2\Delta B_m - t^2}{2(t-1)}$. In Section~\ref{sec:comp-bfv-ckks}, we utilize this expression as Eq.~\ref{EC} to compare the cipher expansion of both schemes.
\end{prop}
\begin{proof}
    For decryption correctness, both schemes must satisfy:
    \begin{itemize}
        \item $q_{\mathsf{MBFV}}:  B^{\mathsf{MP}}_{\mathsf{ct}} < \frac{q}{2t} - \frac{t}{2}$.
        \item $q_{\mathsf{MCKKS}}: \Delta B_m + B^{\mathsf{MP}}_{\mathsf{ct}} < \frac{q}{2}$.
    \end{itemize}
    Comparing both expressions, we have that $q_{\mathsf{BFV}} > q_{\mathsf{CKKS}}$ if:
    \begin{align*}
    &2tB^{\mathsf{MP}}_{\mathsf{ct}} + t^2& & >& &2\Delta B_m + 2B^{\mathsf{MP}}_{\mathsf{ct}} \\
    &2(t -1)B^{\mathsf{MP}}_{\mathsf{ct}}& & >& &2\Delta B_m -t^2 \\
    &B^{\mathsf{MP}}_{\mathsf{ct}}& & >&  &\frac{2\Delta B_m - t^2}{2(t-1)}.
    \end{align*}
\end{proof}

\end{document}